\documentclass{llncs}

\usepackage{citesort}
\usepackage{graphicx}
\usepackage{theorem}
\usepackage{latexsym}
\usepackage{multirow}
\usepackage{algorithm,algorithmic}

\newcounter{Codeline}
\newcommand{\Newcodeline}{\setcounter{Codeline}{1}}
\newcommand{\Cl}{{\theCodeline}: \addtocounter{Codeline}{1}}
\newcommand{\crm}{\\}
\newcommand{\Scomment}[1]{\ensuremath{/\ast} #1 \ensuremath{\ast/}}

\newcommand{\ot}{\ensuremath{\leftarrow}}

\authorrunning{T. Izumi et.al.}
\titlerunning{Minimum Certificate Dispersal with Tree Structures}

\begin{document}

\title{Minimum Certificate Dispersal \\ with Tree Structures
}
\author{Taisuke~IZUMI\inst{1} \and Tomoko~IZUMI\inst{2} \and Hirotaka~ONO\inst{3} \and Koichi~WADA\inst{1}}
\institute{Graduate School of Engineering, Nagoya Institute of Technology, \\ Nagoya, 466-8555, Japan.\\
\email{\{t-izumi, wada\}@nitech.ac.jp}
\and
College of Information Science and Engineering, Ritsumeikan University, \\ Kusatsu, 525-8577 Japan.\\
\email{izumi-t@fc.ritsumei.ac.jp}
\and
Faculty of Economics, Kyushu University, \\ Fukuoka, 812-8581, Japan.\\ 
\email{hirotaka@en.kyushu-u.ac.jp}
} 
\maketitle

\begin{abstract}
Given an $n$-vertex graph $G=(V,E)$ and a set $R\subseteq \{\{x,y\}\mid x,y\in
 V\}$ of requests, we consider to assign a set of edges to each vertex
 in $G$ so that for every request $\{u, v\}$ in $R$ the union of the edge sets 
assigned to $u$ and $v$ contains a path from $u$ to $v$. 
{\em The Minimum Certificate Dispersal Problem} (MCD) is defined as one 
to find an assignment that minimizes the sum of the cardinality of 
the edge set assigned to each vertex. This problem has been shown to 
be LOGAPX-complete for the most general setting, and
 APX-hard and $2$-approximable in polynomial time for dense request
 sets, where $R$ forms a clique. In this  
 paper, we investigate the complexity of MCD with sparse (tree)
 structures. We first show that MCD is APX-hard when $R$ is a
 tree, even a star. We then explore the problem from the viewpoint of
 the {\em maximum degree} $\Delta$ of the tree: MCD for tree request  
set with constant $\Delta$ is solvable in polynomial time, while 
that with $\Delta=\Omega(n)$ is $2.56$-approximable in polynomial time
 but hard to approximate within $1.01$ unless P$=$NP. As for the 
structure of $G$ itself, we show that the problem can be 
 solved in polynomial time if $G$ is a tree. 
\end{abstract}

\section{Introduction}
\noindent{\bf Background and Motivation.}\ \ 
Let $G = (V,E)$ be a graph and $R\subseteq \{\{x,y\}\mid x,y\in
V\}$ be a set of pairs of vertices, which represents requests about
reachability between two vertices.  For given $G$ and $R$, we consider an assignment of a set of edges to each vertex in $G$. 
The assignment satisfies a request $\{u, v\}$ if the union of the edge sets assigned to $u$ and $v$ 
contains a path from $u$ to $v$. The {\em Minimum Certificate Dispersal Problem (MCD)} 
is the one to find the assignment satisfying all requests in $R$ that minimizes the sum of 
the cardinality of the edge set assigned to each vertex.

This problem is motivated by a requirement in public-key based security
systems, which are known as a major technique for supporting secure communication 
in a distributed system \cite{Capkun03,Gouda04,Gouda05,Hubaux01,Izumi10,Jung04,Zheng05,Zheng06}. 
The main problem of the
systems is to make each user's public key available to others in such a way
that its authenticity is verifiable. One of well-known approaches to solve this
problem is based on public-key certificates. A public-key certificate contains
the public key of a user $v$ encrypted by using the private key of another user $u$. If
a user $u$ knows the public key of another user $v$, user $u$ can issue a certificate
from $u$ to $v$. Any user who knows the public key of $u$ can use it to decrypt the
certificate from $u$ to $v$ for obtaining the public key of $v$. All certificates issued
by users in a network can be represented by a certificate graph: Each vertex
corresponds to a user and each directed edge corresponds to a certificate. When
a user $w$ has communication request to send messages to a user $v$ securely, $w$
needs to know the public key of $v$ to encrypt the messages with it. For satisfying
a communication request from a vertex $w$ to $v$, vertex $w$ needs to get vertex
$v$'s public-key. To compute $v$'s public-key, $w$ uses a set of certificates stored in
$w$ and $v$ in advance. Therefore, in a certificate graph, if a set of certificates
stored in $w$ and $v$ contains a path from $w$ to $v$, then the communication request
from $w$ to $v$ is satisfied. In terms of cost to maintain certificates, the total
number of certificates stored in all vertices must be minimized for satisfying all
communication requests.

The previous work mainly focuses on directed variants of MCD, in
which graph $G$ is directed. Jung et al. discussed MCD with a
restriction of available paths in \cite{Jung04} and proved that  
the problem is NP-hard. In their work, to assign edges to each vertex,
only the restricted paths that are given for each request is allowed to
be used. MCD with no restriction about available paths was first 
formulated in \cite{Zheng06}. This variant is also 
proved to be NP-hard even if the input graph is a strongly connected
directed graph. On the other hand, MCD for directed graphs with $R$ forming
a clique is polynomially solvable for bidirectional trees and rings, 
and Cartesian products of graphs such as meshes and hypercubes~\cite{Zheng06}. 

Based on these work, the (in)approximability of MCD for directed graphs 
has been studied from the viewpoint of the topological structure of 
$R$ (not $G$)~\cite{Izumi10}. Since MCD is doubly structured (one is
the graph $G$ itself and the other is the request structure $R$), 
the hardness of MCD depends not only on the topology of $G$ but also on
the one of $R$. In view of these observation, the (in)approximability of 
MCD for directed graphs is investigated for general case and $R$ forming
a clique, as a typical community structure. It was shown that the former
case is $O(\log |V|)$-approximable in polynomial time but has no
polynomial time algorithm whose approximation factor is better than
$0.2266 \log |V|$ unless P$=$NP. The latter case is $2$-approximable 
but has no polynomial time algorithm whose approximation factor is
better than $1.001$, unless P$=$NP. In \cite{Izumi10}, 
the undirected variant of MCD is also considered, and 
$1.5$-approximation algorithm for the case when $R$ forming a clique is
presented. 

These results naturally raise a new question:
For the hardness of approximation or constant-factor approximability, is
such a dense structure (i.e., clique) essential? For example, how is the
case when $R$ is sparse, e.g., a tree? 
This paper further investigates the (in)approximability of MCD when 
$R$ forms a tree, as another typical topology.



\noindent{\bf Our Contribution.}\ \ 
We investigate the complexity of MCD with tree structure. 
Here, we say ``with tree structure'' in two senses. 
One is the case when $R$ forms a tree, and the other is the case when
$G$ itself is a tree. 
One reason of this focus has already been mentioned above. 
Another reason is that a tree is a minimal connected structure; 
even if $G$ (resp., $R$) is not a tree, by solving MCD for $G'$, a
spanning tree of $G$ (resp., for a spanning tree $R'$ of $R$), 
we can obtain an upper bound on the optimal solution (resp., a lower bound
on the optimal solution) of the original MCD problem.  

For MCD with tree $R$, we show that the hardness and  
approximability depend on the {\em maximum degree} $\Delta$ of tree $R$:
MCD for tree $R$ with constant degree is solvable in polynomial time
while that with $\Omega(n)$ degree is APX-hard. 
As for MCD for tree $G$, we present a polynomial optimal algorithm. The
followings are summary of our contributions:  
\begin{itemize}
\item {\em $R$ is an arbitrary tree}: First we consider
      MCD for the 
      case when $R$ is a {\em star}. 
      Even in this simplest setting, MCD is shown to be APX-hard: MCD
      for undirected graph $G$ with sparse $R$ is still APX-hard. 
      Moreover, the reduction {\em to}
      the Steiner tree problem for unweighted graphs(STREE) leads 
      to an upper bound 1.28 on approximation
      ratio for MCD with star request sets. For 
      arbitrary tree $R$, it is shown that there 
      is a 2.56-approximate algorithm for MCD by utilizing the
      approximation algorithm for star $R$. 
\item {\em $R$ is a tree with $\Delta=O(\log |V|)$}: 
By using a similar analysis to arbitrary tree $R$, the upper bound of
      approximation ratio for MCD can be reduced to 2. In particular,
      if $R$ is a star with $\Delta = O(\log n)$ MCD is polynomially solvable.
\item {\em $R$ is a tree with constant degree}: This case is
      polynomially solvable. These imply that the hardness of MCD for
      tree $R$ heavily depends on its maximum degree. A key idea is 
      to define normal solutions. Our dynamic programming based 
      algorithm searches not the whole solution space but (much smaller)
      normal solution space. 
\item {\em $G$ is an arbitrary tree}: In this case also, a positive
      result is shown. For any request set $R$ (not restricted to a
      tree), our algorithm outputs an optimal solution in polynomial
      time. The algorithm exploits the polynomial time solvability of  
      VERTEX-COVER for bipartite graphs. 
\end{itemize}

\medskip

The remainder of the paper is organized as follows. In Section
\ref{sec:model}, we formally define the Minimum Certificate Dispersal
Problem (MCD). Section \ref{sec:star} shows 
the hardness and approximability of MCD with star request sets, 
and Section \ref{sec:tree} extends it to the approximability 
of MCD with tree request sets. In Section \ref{sec:consttree}, we present
a polynomial time algorithm that optimally solves MCD for tree request 
with constant degree. Section \ref{sec:graph} shows an optimal algorithm
for MCD with undirected tree graphs. Section \ref{sec:conclusion}
concludes the paper. 

\section{Minimum Certificate Dispersal Problem}\label{sec:model}

While the Minimum Certificate Dispersal (MCD) Problem is originally
defined for directed graphs, we deal with its undirected variant,
where the given graph is undirected. 
The difference between them is the meaning of assignment an edge to 
a vertex: In the standard MCD, an edge $(u, v)$ means
a certificate from $u$ to $v$. In the undirected variant of MCD, edge means a
bidirectional certificate from $u$ to $v$ and $v$ to $u$ which is
not separable. Since we treat the undirected variants of MCD throughout this paper, 
we simply refer those problems as MCD. In the following, we give the 
formal definition of MCD problem.


Let $G=(V, E)$ be an undirected graph, where $V$ and $E$ are the sets of
vertices and edges in $G$, respectively. An edge in $E$ connects two distinct
vertices in $V$. The edge between vertex $u$ and $v$ is denoted by $\{u, v\}$.
The numbers of vertices and edges in $G$ are denoted by $n$ and $m$, 
respectively (i.e., $n=|V|, m=|E|$). 
A sequence of edges $p(v_0, v_k) = \{v_0, v_1\}, \{v_1, v_2\}, \dots , \{v_{k-1}, v_k\}$ 
is called a {\em path} from $v_0$ to $v_k$ of length $k$.
A path $p(v_0, v_k)$ can be represented by a sequence of vertices
$p(v_0, v_k) = (v_0, v_1,\dots , v_k)$. 
For a path $p(v_0, v_k)$, $v_0$ and $v_k$ are called the endpoints 
of the path. 
A shortest path from $u$ to $v$ is the one
whose length is the minimum of all paths from $u$ to $v$, 
and the distance from $u$ to $v$ is the length of a shortest path from $u$ to $v$,
denoted by $d(u, v)$.

A {\em dispersal} $D$ of an undirected graph $G=(V, E)$ is a family of sets of
edges indexed by $V$, that is, $D=\{ D_v \subseteq E \mid v \in V \} $. 
We call $D_v$ a local dispersal of $v$. A local dispersal $D_v$ indicates 
the set of edges assigned to $v$. The {\em cost} of a dispersal $D$, denoted by
$c(D)$, is the sum of the cardinalities of all local dispersals in $D$
(i.e., $c(D) = \Sigma _{v \in V} |D_v|$). 
A request is a reachable unordered pair of vertices in $G$.
For a request $\{u, v\}$, $u$ and $v$ are called the endpoints
of the request.  
We say a dispersal $D$ of $G$ {\em satisfies} 
a set $R$ of requests if a path between $u$ and $v$ is included in $D_u \cup D_v$
for any request $\{u, v\} \in R$. Given two dispersals $D$ and $D'$ of $G$,
the union of two dispersals $\{D_v \cup D'_v \mid v \in V \}$ is denoted by
$D \cup D'$.

The {\em Minimum Certificate Dispersal Problem (MCD)} is defined as follows:
\begin{definition}[Minimum Certificate Dispersal Problem (MCD)]
\hspace*{1mm} \\
INPUT: An undirected graph $G=(V, E)$ and a set $R$ of requests.\\
OUTPUT: A dispersal $D$ of $G$ satisfying $R$ with minimum cost.
\end{definition}

The minimum among costs of dispersals of $G$ that satisfy $R$ is 
denoted by $c_{min}(G, R)$.
Let $D^{Opt}$ be an optimal dispersal of $G$ which satisfies $R$
(i.e., $D^{Opt}$ is one such that $c(D^{Opt}) = c_{min}(G, R)$).

Since $R$ is a set of unordered pairs of $V$, it naturally defines an undirected graph $H_R
= (V_R, E_R)$ where
$V_R = \{ u, v \mid  \{u, v\} \in R \}$ and $E_R =  R$.
The request set $R$ is called {\em tree} if $H_R$ is a tree, and is 
also called {\em star} if it is a tree with exactly one internal vertex.
The maximum degree of $H_R$ is denoted by $\Delta_R$. The problem of MCD restricting $H_R$ to tree or star with degree $\Delta_R$ is called
MCD-tree($\Delta_R$) and MCD-star($\Delta_R$). We also denote the problem 
of MCD restricting $H_R$ to tree (or star) with degree $\Delta_R=O(f(n))$ for some function $f(n)$ as MCD-tree($O(f(n))$) (or MCD-star($O(f(n))$). When we do not consider 
any constraint to the maximum degree, the argument $\Delta_R$ is omitted.


\section{MCD for Star Request Sets}
\label{sec:star}

The NP-hardness and inapproximability of directed MCD for strongly-connected 
graphs are shown in the previous work\cite{Zheng06}. In this section, we prove that 
MCD is APX-hard even if we assume that $H_R$ is a star.
The proof is by the reduction from/to the Steiner-tree problem. 
\begin{definition}[Steiner-tree Problem (STREE)]
\hspace*{1mm} \\
INPUT: An undirected connected graph $G=(V, E)$ and a set 
$T \subseteq V$ of terminals.\\
OUTPUT: A minimum-cardinality subset of edges $E' \subseteq E$ that connects 
all terminals in $T$.
\end{definition}

We often use the notation STREE($t$) and STREE($O(f(n))$), which are
the Steiner-tree problems for a terminal set with cardinality at most $t$
and $t = O(f(n))$ respectively. 

\begin{theorem} \label{theorem:star-dispersal}
There exists a polynomial time $\rho$-approximation algorithm 
for MCD-star($\Delta$) if and only if there exists a polynomial time 
$\rho$-approximation algorithm for STREE($\Delta+1$). 
\end{theorem}

\begin{proof}
We prove the only-if part and if part can be
proved in almost the same way as the proof of the only-if part.
Given an instance $(G=(V,E), T)$ of STREE($t+1$), we construct 
an instance $(G' , R)$ of MCD-star($t$) as $G = G'$ and 
$R = \{ \{v_r, u\} \mid u \in T \setminus \{v_r\} \}$, 
where $t = \Delta_R$ and $v_r$ is an arbitrary vertex in $T$.
To prove the theorem, it suffices to show that any feasible solution 
of MCD $(G', R)$ (resp. $(G, T)$) can be transformed to a
feasible solution of $(G, T)$ (resp. $(G', R)$) with no gain of 
solution cost. Then because  $(G', R)$ and $(G, T)$ have
the same optimal cost and thus any $\rho$-approximated 
solution of $(G', R)$ induces an $\rho$-approximated solution of 
$(G, T)$.

\noindent
{\bf From  MCD-star($\Delta$) to STREE($\Delta+1$)}: Given a feasible 
solution $D = \{D_v \mid v \in V\}$ of $(G', R)$, 
we can construct a feasible solution $S = \cup_{v\in V} D_v$ of STREE. 
Since $S$ necessarily includes a path between any pair in $R$, its induced graph 
is connected and contains  all vertices in $T = V_R$. Thus,  $S$ is a feasible solution for STREE and its cost 
is at most $\sum_{v \in V} |D_i|$. 

\noindent
{\bf From STREE($\Delta+1$) to MCD-star($\Delta$)}: 
Given a feasible solution
$S$ of $(G, T)$, we obtain the solution of MCD-star by assigning 
all edges in $S (\subseteq E)$ to the internal vertex $v_r$ of $H_R$. 
Since $D_{v_r}$ connects all vertices in $V_R$, 
any request in $R$ is satisfied. Thus $D = \{D_{v_r}=S\} \cup \{
 D_v=\emptyset \mid v \in V, v \neq v_r\}$ is a feasible solution of $(G, R)$ and its cost 
is equal to $|S|$.

Then since MCD-star($\Delta$) and STREE($\Delta+1$) have  the same optimal cost, the theorem is proved. \qed
\end{proof}

Since STREE is APX-hard \cite{Bern89} and its known upper and lower bounds 
for the approximation factor are 1.28 and 1.01, respectively \cite{Robin00,Chlebik08}, we can 
obtain the following corollary.

\begin{corollary} \label{corol:MCD-star}
MCD-star is APX-hard, has a polynomial time $1.28$-approximation algorithm,
and has no polynomial time algorithm with an approximation factor 
less than $1.01$ unless $P = NP$.
\end{corollary}

\section{MCD for Tree Request Sets}
\label{sec:tree}

\subsection{Tree Structure with Arbitrary Degree}

The general approximability of MCD-tree can be shown by the following
theorem:

\begin{theorem}\label{theorem:tree}
Provided any $\rho$-approximation algorithm for MCD-star,
there is a polynomial time $2\rho$-approximation algorithm for MCD-tree.
\end{theorem}

We first introduce the construction of the algorithm: Given an instance 
($G=(V,E), R$) of MCD-tree, we regard $H_R$ as a rooted tree by 
picking up an arbitrary vertex as its root. Letting $\mathit{depth}(v)$ 
($v \in V_R$) be the distance from the root to $v$ on $H_R$, we 
partition the request set $R$ into two disjoint 
subsets $R^{i}$ ($i \in \{0, 1\}$) as $R^{i} = \{\{u, v\}\mid  \mathit{depth}(u) 
< \mathit{depth}(v) \ \mathrm{and} \  \mathit{depth}(u) \bmod 2 = i\}$. Note that 
both $R^{1}$ and $R^{0}$ respectively form two forests where each 
connected component is a star. Thus, using any algorithm for MCD-star 
(denoted by $\mathcal{A}$), we can obtain two solutions of 
$(G, R^{1})$ and $(G, R^{0})$ by independently solving the 
problems associated with each connected component. Letting $D^j$ 
be the solution of instance $(G, R^{j})$, the
union $D^1 \cup D^0$ is the final solution of our algorithm.

It is obvious that the returned solution is feasible. Since 
both of $c(D^1)$ and $c(D^0)$ are the lower bound of
the optimal cost for $(G, R)$, the algorithm achieves 
approximation ratio $2\rho$. For lack of the space,
we give the proof details in the appendix. The above theorem and 
Corollary \ref{corol:MCD-star} leads the following corollary: 
\begin{corollary}
MCD-tree has a polynomial time $2.56$-approximation algorithm.
\end{corollary}

\subsection{Tree Structures with $O(\log n)$ Degree}

\sloppy{
In the proof of Theorem \ref{theorem:tree}, we have shown that 
the approximated solution for instance $(G, R)$ of MCD-tree can be 
constructed by solving several MCD-star instances. Thus, 
if $\Delta_R = O(\log n)$, each decomposed star has $O(\log n)$ vertices
(that is, an instance of MCD-star($O(\log n)$)). By Theorem 
\ref{theorem:star-dispersal}, MCD-star($O(\log n)$) and STREE($O(\log n)$)
have the same complexity and STREE($O(\log n)$) is optimally solved in
polynomial time \cite{Dreyfus72}. Therefore, 
Theorem\ref{theorem:tree} leads the following corollary. 
 }
\begin{corollary}
There is an optimal algorithm to solve MCD-star($O(\log n)$) in polynomial time and
there is  an approximation factor 2 polynomial time algorithm 
for MCD-tree($O(\log n)$).
\end{corollary}

\section{Tree Structures with Constant Degree}
\label{sec:consttree}
\vspace*{-.3cm}

In this section, we provide an algorithm that returns the optimal dispersal
for any instance of MCD-tree($O(1)$). 
Throughout this section, we regard $H_R$ as a rooted tree by picking 
up an arbitrary vertex $r$ in $V_R$ as its root. Given a vertex $u \in V_R$,
let $\mathit{par}(u)$ be the parent of $u$, and let 
$\mathit{Child}(u)$ be the set of $u$'s children. 

A request $\{u, v\}$ is {\em well-satisfied} by a feasible $D$ 
if there exists a vertex $\alpha_{u, v}$ such that $D_u$ contains a path from 
$u$ to $\alpha_{u,v}$ and $D_v$ contains a path from $\alpha_{u,v}$ to $v$.
Then, vertex $\alpha_{u, v}$ is called the {\em connecting point} of request 
$\{u, v\}$ in $D$.

We begin with the following fundamental
property:

\begin{lemma} \label{lemma:connectivity}
For any instance $(G, R)$ of MCD-tree,
there is an optimal solution that well-satisfies all requests in $R$.
\end{lemma}

By the above lemma, we can reduce the search space to one where
each feasible solution well-satisfies all requests. In the following 
argument, we assume that every request has a connecting point 
in the optimal dispersal. The principle of our algorithm is to 
determine the connecting points 
recursively from the leaf side of $H_R$ via dynamic programming. 
Let $T_R(u) = (V_R(u), E_R(u))$ be the subtree of $H_R$ rooted 
by $u$, $D^{\ast}(u, \alpha)$ be a dispersal for 
instance $(G, E_R(u))$ with the smallest cost such that 
$D_u$ contains a path to from $u$ to $\alpha$. Note that 
$D^{\ast}(r, r)$ is an optimal solution of $(G, R)$.
We define $\gamma(u) = |\mathit{Child}(u)|$ for short. The key recurrence 
of the dynamic programming can be stated by the following lemma:

\begin{lemma} \label{lemma:reccursion}
Let $u$ and $\alpha$ be vertices in $V$ and let 
$A=(\alpha_1,...,\alpha_{\gamma(u)}) \in V^{\gamma(u)}$. Then 
the following equality holds:
\[
c(D^{\ast}(u, \alpha)) = 
\min_{A \in V^{|\gamma(u)|}} \{ 
c(D^{\mathit{Opt}}(G, E_A \cup \{\{u, \alpha\}\})) + 
\sum_{u_k \in \mathit{Child}(u)} c(D^{\ast}(u_k, \alpha_k))\}
\]
where $E_A = \{\{u, \alpha_1\}, \{u, \alpha_2\}, \cdots, 
\{u, \alpha_{\gamma(u)}\}\}$.
\end{lemma}

This recurrence naturally induces a polynomial time algorithm 
for MCD-tree($O(1)$). The pseudo-code of the algorithm is shown in
Algorithm \ref{algo:consttree}. The algorithm maintains a table $D^{\ast}$,
where each entry $D^{\ast}[u][\alpha]$ stores the solution 
$D^{\ast}(u, \alpha)$. The core of the algorithm is to fill
the table following the recurrence of Lemma \ref{lemma:reccursion}:
Assume an arbitrary ordering $\sigma = u_1, u_2, \cdots u_{|V_R|}$ 
of vertices in $V_R$ where any vertex appears after all of its descendants 
have appeared. To compute the solution to be stored 
in $D^{\ast}[u_i][\alpha]$, the algorithm considers all possible 
choices of connecting points to $u_i$'s children. Let 
$q_1, q_2, \cdots q_{\gamma(u_i)}$ be the children of $u_i$. 
Fixing a choice $A = (\alpha_{u_i, q_1}, \alpha_{u_i, q_2}, 
\cdots \alpha_{u_i, q_\gamma(u_i)})$ 
of connecting points (in the pseudo-code, $\alpha_k$ corresponds to 
$\alpha_{u_i, q_k}$), the algorithm determines the local dispersal to 
$u$ by computing the optimal solution for $(G, E_A \cup \{\{u_i, \alpha\}\})$. 
Note that this can be computed in polynomial time because the request
set forms a constant-degree star. By Theorem \ref{theorem:star-dispersal},
it is equivalent to STREE($O(1)$). Letting $D'$ be the computed solution for 
$(G, E_A \cup \{\{u_i, \alpha\}\})$. we obtain 
$D = D' \cup D^{\ast}[q_1][\alpha_{u_i, q_1}] \cup 
D^{\ast}[q_2][\alpha_{u_i, q_2}] \cup 
\cdots \cup D^{\ast}[q_{\gamma(u)}][\alpha_{u_i, q_{\gamma(u)}}]$. Importantly,
we can assume that only $D'_u$ is nonempty in $D'$ (recall 
the construction of MCD-star solutions from STREE solutions), which
implies that $D_{u_i}$ has a path to any connecting point $\alpha_{u_i, q_j}$ 
in $A$. Since it $D_{q_j}$ has a path from $q_j$ to $\alpha_{u_i, q_j}$ from 
the definition of $^{\ast}[q_i][\alpha_{u_i, q_j}]$, $D_{u_i} \cup D_{q_j}$
necessarily has the path between $u_i$ and $q_j$, That is the feasibility of 
$D$ is guaranteed. If $D$ is better than the solution already computed 
(for other choice of $A$), $D^{\ast}[u_i][\alpha]$ is updated by $D$. 
After the computation for all possible choices of $A$, 
$D^{\ast}[u_i][\alpha]$ stores the optimal solution. Finally, after filling 
all entries of the table, the algorithm returns 
$D^{\ast}[u_{|V_R|}][u_{|V_R|}]$, which is the 
optimal solution for instance $(G, R)$. 

Lemma \ref{lemma:reccursion} obviously derives the correctness of 
Algorithm \ref{algo:consttree}. Since we assume that the maximum 
degree of tree $H_R$ is a constant, the number of tuples of $A$
is also a constant. Thus the number of possible choices about $A$ is bounded 
by a polynomial of $n$. It follows that the running time of Algorithm 
\ref{algo:consttree} is bounded by a polynomial of $n$. We can 
have the following theorem:

\begin{theorem}
There is a polynomial time algorithm solving MCD-tree($O(1)$).
\end{theorem}
\vspace*{-.3cm}

\Newcodeline
\begin{algorithm}[h]
\caption{Polynomial Time Algorithm for MCD-tree($O(1)$)}
\label{algo:consttree}
{\small 
\begin{tabbing}
111 \= 11 \= 11 \= 11 \= 11 \= 11 \= 11 \= \kill
\Cl \> $D^{\ast}[V_R][V]$ : the array storing the computed solutions \crm
\Cl \> \> (All entries are initialized by a dummy solution with cost $\infty$)
\crm
\Cl \> $\sigma = u_1, u_2, \cdots u_{|V_R|}$ : an ordering of $V_R$ \crm 
\Cl \> \> containing parent-child relationship on $H_R$ 
(children come earlier).\crm
\> \crm
\Cl \> {\bf for each} $u_i \in V_R$ in order of $\sigma$ {\bf do} \crm 
\Cl \> \> Let $Q = (q_1, q_2, \cdots q_{\gamma(u_i)})$ be an arbitrary
ordering of $\mathit{Child}(u_i)$ \crm
\Cl \> \> {\bf for each} $(A, \alpha) = 
((\alpha_1, \alpha_2, \cdots, \alpha_{\gamma(u_i)}), \alpha) 
\in V^{\gamma(u_i)} \times V$ {\bf do} \crm
\Cl \> \> \> $D' \ot$ the optimal solution of $(G, E_A \cup \{\{u, \alpha\}\})$
s.t. only $D'_{u_i}$ is nonempty. \crm
\Cl \> \> \> \Scomment{$E_A = \{\{u_i, \alpha_1\}, \{u_i, \alpha_2\}, \cdots
\{u_i, \alpha_{\gamma(u_i)}\}\}$} \crm
\Cl \> \> \> $D \ot D' \cup \left( \bigcup_{j \in [1, \gamma(u_i)]} 
D^{\ast}[q_j][\alpha_j] \right)$ \crm
\Cl \> \> \> {\bf if} $c(D^{\ast}[u_i][\alpha]) > c(D)$ {\bf then} 
$D^{\ast}[u_i][\alpha] \ot D$ \crm
\Cl \> \> {\bf endfor} \crm
\Cl \> {\bf endfor} \crm 
\Cl \> return $D^{\ast}[u_{|V_R|}][u_{|V_R|}]$
\end{tabbing}
}
\end{algorithm}

\vspace*{-.5cm}
\section{MCD for Tree Graphs}
\label{sec:graph}
\vspace*{-.3cm}

While the previous sections focus on the structure of $H_R$, in this section, 
we look at the structure of graph $G$: We show that MCD is solvable in
polynomial time if $G$ is a tree.
In the algorithm, we compute for each edge $e\in E$ which $D_u$ should
contain $e$; for each $e\in E$,  we decide $\{u\in V \mid e\in D_u \}$. 
For this decision about $e\in E$, we utilize a bipartite graph 
that represents whether a request $\{u,v\}$ should use $e$ in its path. 

Let $T=(V, E)$ be a tree and $R$ be a request set. Now we consider to
decide $\{u\in V \mid e\in D_u \}$ for an edge $e=\{u,v\}\in E$. 
By deleting $e=\{u,v\}$ from $T$, we obtain two subtrees $T_u=(V_u,
E_u)$ and $T_v=(V_v, E_v)$ of $T$, where $T_u$ and $T_v$ contain $u$ and
$v$, respectively. Note that $V_u\cap V_v = \emptyset$ and $V=V_u \cup
V_v$.  From these two subtrees $T_u$ and $T_v$, we construct a bipartite
graph $B_{uv} = (V_u \cup V_v, E_{uv})$, where $E_{uv} = \{ \{ a, b \}
\in R \mid a \in V_u, b\in V_v \}$. 
It should be noted that $\{e\}$ is an $a$-$b$ cut for every $\{a,b\}\in
E_{uv}$, since $T$ is a tree. Thus, this bipartite graph represents  
that if an edge $\{w_i,w_j\}\in E_{uv}$, at least one of $w_i$ or $w_j$
should have $e=\{u,v\}$ in its local dispersal, i.e., $e \in D_u \cup
D_v$, otherwise $D$ does not satisfy request $\{w_i,w_j\}$ due to cut $\{e\}$. 

This condition is interpreted as a vertex cover of $B_{uv}$. 
A {\em vertex cover} $C$ of a graph is a set of vertices
such that each edge in its edge set is incident to at least one vertex in
$C$. 
Namely, a necessary condition of $D$ satisfying $R$ is that 
for each $e=\{u,v\}$, $C_{uv}=\{w \in V \mid e \in D_w \}$ is a vertex
cover of $B_{uv}$. We call this {\em vertex cover condition}. 
It can be shown that the vertex cover condition is also
sufficient for $D$ to satisfy $R$. Suppose that a dispersal $D$
satisfies the vertex cover condition. For a request $\{a_0,a_k\}$
and its unique path $p(a_0,a_k)=(a_0,a_1,\ldots,a_k)$ on $T$, by the
definition of $B_{uv}$, every $B_{a_{i}a_{i+1}}$  contains
edge $\{a_{0},a_{k}\}$. By the vertex cover condition, $\{a_{i},a_{i+1}\}\in
D_{a_0} \cup D_{a_k}$ holds for $i=0,\ldots,k-1$, which implies 
$D_{a_0}\cup D_{a_k}$ contains path $p(a_0,a_k)$; $D$ satisfies request
$\{a_0,a_k\}$. 

By these arguments, the vertex cover condition is equivalent to the
feasibility of $D$. Also it can be seen that choices of vertex cover of
$B_{uv}$ and another $B_{u'v'}$ are independent to each other in terms
of the feasibility of $D$. These imply that the union of the minimum
size of vertex cover for $B_{uv}$'s is an optimal solution of MCD for
tree $G$.  

From these, we obtain the following algorithm: For every edge $\{ u, v
\} $ in $T$, we first compute a minimum vertex cover $C_{uv}$ of 
bipartite graph $B_{uv}$. Then, let $D_w = \{ \{ u, v \} \in E \mid w \in
C_{uv} \}$ and output. Since VERTEX-COVER problem for bipartite graphs
can be solved via the {\em maximum matching problem}~\cite{Konig31},  
whose time complexity is $O(\sqrt{n}m)$ time, where $n$ and $m$ are the
numbers of vertices and edges, respectively~\cite{Hopcroft73}. Thus, 
MCD for undirected tree $G$ can be solved in $O(n^{1.5}|R|)$ time. 

\begin{theorem}
For an undirected tree graph $G$ and any request $R$, 
MCD is solvable in $O(n^{1.5}|R|)$ time. 
\end{theorem}

\vspace*{-.7cm}
\section{Concluding remarks}
\label{sec:conclusion}
\vspace*{-.3cm}
We have considered undirected variants of the MCD problem with tree structures
and shown that for MCD with tree $R$, the hardness and approximability depend on
the maximum degree of tree $R$ and MCD for any $R$ can 
be solved in polynomial time if $G$ is a tree.

There are interesting open problems as follows;
\begin{itemize}
\item The hardness of MCD-tree($O(\log n)$): Even NP-hardness of that 
class is not proved yet. Precisely, no hardness result is found
for MCD-tree($\Delta_R$) where $\Delta_R=o(n)$ and $\Delta_R=\omega(1)$.
\item The graph class of $G$ allowing any request set $R$ to be tractable:
The case of trees (shown in this paper) is only the known class making the problem solvable in polynomial time. We would like to know what sparse 
graph classes (e.g., rings, series-parallel graphs, and planar graphs) 
can be solved for any request $R$ in polynomial time. In particular, for 
MCD of rings with any request $R$ we would like to decide whether it is 
NP-hard or P.
\item Related to the question right above, we would like to extend
      the DP technique for MCD-tree($O(1)$) presented in Section 
      \ref{sec:consttree} to other wider classes of $H_R$. 
      Some sparse and degree-bounded graphs might be its
      candidates. In fact, the key of polynomial time running time of 
      Algorithm \ref{algo:consttree} is based only on the following two 
      conditions: (1) There exists an optimal solution that
      well-satisfies $R$, (2) There exists an ordering $\sigma$ on $V_R$ 
      such that every cut $(\{\sigma(1),\ldots,\sigma(i)\},
      \{\sigma(i+1),\ldots, \sigma(|V_R|)\})$ on $H_R$ has a constant
      size.  
\item The complexity gap between undirected MCD and directed MCD:
In general, directed MCD is not easier than undirected MCD in the 
sense that the latter is a special case of the former.
But it is unknown whether it is proper or not. It is not quite trivial 
to transform any known complexity result for MCD into directed MCD, and 
vice versa.
\end{itemize}

\vspace*{-.5cm}
\bibliographystyle{abbrv}

\newpage
\appendix

\section{Omitted Proof}

\subsection{The proof of Theorem \ref{theorem:tree}}

\begin{proof}
Let $\mathit{Opt}(G, R)$ be an optimal solution of $(G, R)$, and
$\mathcal{A}(G, R)$ be the solution of $(G, R)$ returned by algorithm
$\mathcal{A}$. Installing $\rho$-approximation algorithm of MCD-star
into $\mathcal{A}$, we can obtain $\rho$-approximated solutions 
of $(G, R^{1})$ and $(G, R^{0})$ because each connected component 
of $V_R^{1}$ and $V_R^{0}$ is a star (trivially, the set of 
$\rho$-approximated solutions corresponding to each connected components
induces an $\rho$-approximated solution of the whole instance). 
Thus, we have $c(\mathcal{A}(G, R^{j})) \leq \alpha c(\mathit{Opt}(G, R^{j}))$
($j \in \{0, 1\})$. Furthermore, since $R^{j} \subseteq R$ holds
for any $j \in \{0, 1\}$, we also have $c(\mathit{Opt}(G, R^{j})) \leq c(\mathit{Opt}(G, R))$. Letting $S$ be the solution of $(G, R)$ finally 
returned and $c_\mathrm{max} = \max\{c(\mathit{Opt}(G, R^{1})), c(\mathit{Opt}(G,
R^{0}))\}$, we finally obtain $c(S) \leq c(\mathcal{A}(G, R^{1}))) + 
c(\mathcal{A}(G, R^{0})) \leq 2\alpha c_\mathrm{max} \leq 2\alpha
\mathit{Opt}(G, R)$. The theorem is proved. \qed
\end{proof}

\subsection{The proof of Lemma \ref{lemma:connectivity}}

\begin{proof}
The proof is done in a constructive way. That is, we show that 
it is possible to transform any optimal solution to one well-satisfying 
all requests with no extra cost. Let $D$ be an optimal solution, 
$U$ be the set of vertices having at least one request not well-satisfied, 
and $u$ be the vertex farthest from $r$ in $U$.  Since $u$ is the farthest, 
only the request between $u$ and its parent is not well-satisfied
in all requests related to $u$. Let $v=\mathit{par}(u)$ for short. To prove 
the lemma, it suffices to show that we can obtain a solution 
$D'$ where $c(D) = c(D')$ holds, any request well-satisfied in $D$ is 
also done in $D'$, and $\{u, v\}$ is well-satisfied. 
Let $e_0, e_1, \cdots e_k$ be the sequence of edges in $G$ 
organizing a path from $u$ to $v$. From the fact that $\{u, v\}$ 
is not well-satisfied, there exists an edge $e_j \in D_{u}$ such that 
$e_{p} \in D_{v}$ for some $p < j$ and $e_{q} \in D_{v}$ for any $q > j$.
Since request $\{u, u'\}$ is well-satisfied for any $u' \in \mathit{Child}(u)$,
there is a path $P_{u'}$ in $D_u$ from $u$ to $\alpha_{u, u'}$ in $D_u$. 
Then, for any $u' \in \mathit{Child}(u)$, each $P_{u'}$ does not contain $e_j$
because $\{u, v\}$ becomes well-satisfied if $e_j \in P_{u'}$ holds for 
some $u'$ (see Figure \ref{fig:connectivity}). Thus, we can construct 
a dispersal $D'$ as $D'_{x} = D_{x}$ for any $x \neq u, v$, 
$D'_{u} = D_{u} \setminus \{e_j\}$ and  $D'_{v} = D_{v} \cup \{e_j\}$, 
which is feasible and has the same cost as $D$. Repeating 
the construction, we can make request $\{u, v\}$ well-satisfied. Since 
this procedure does not break the well-satisfied property of any other 
request, we can eventually obtain a feasible solution well-satisfying
all requests without extra cost. The lemma is proved. \qed
\end{proof}

\begin{figure}[t]
\begin{center}
\includegraphics[keepaspectratio,width=85mm]{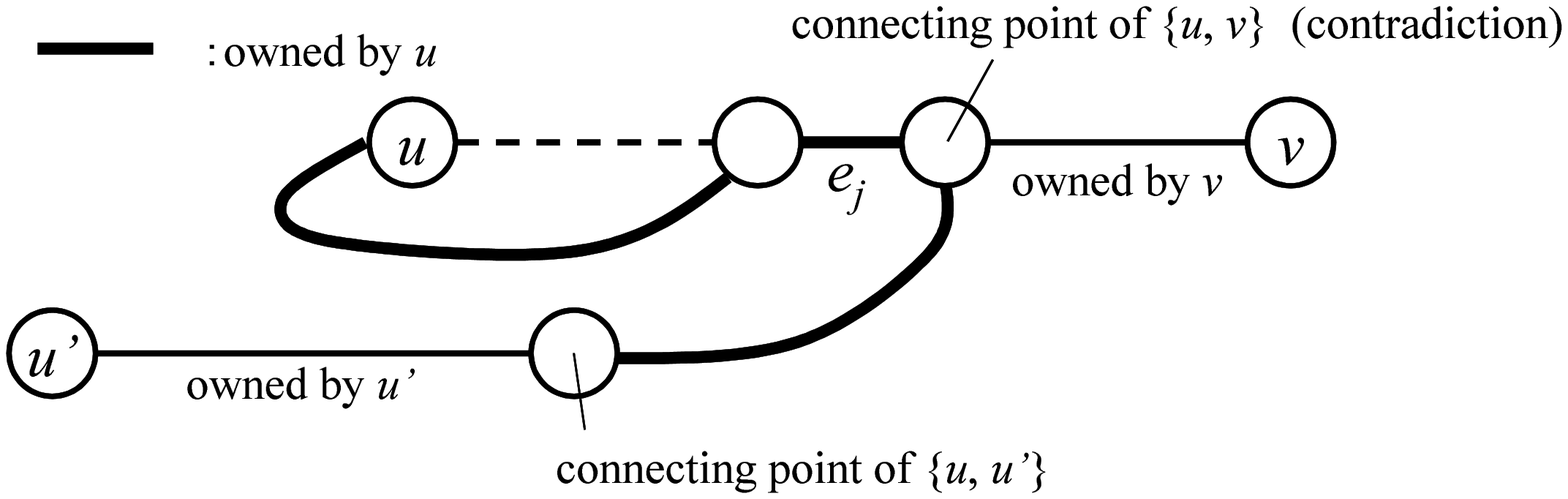}
\caption{Illustration of the proof of Lemma \ref{lemma:connectivity}:
If $e_j \in D_u$ is used to satisfy request $\{u, u'\}$, $D_u$ contains
a path terminating with $e_j$ because $\{u, u'\}$ is well-satisfied. It
follows that request $\{u, v\}$ becomes well-satisfied, which is a
contradiction.}
\label{fig:connectivity}
\end{center}
\end{figure}

\subsection{The Proof of Lemma \ref{lemma:reccursion}}

\begin{proof}
Let $\alpha'_k$ be the connecting point of $\{u, u_k\}$ in 
$\mathit{D}^{\ast}(u, \alpha)$, and $D^{\ast}_u$ be $u$'s
local dispersal in $D^{\ast}(u, \alpha)$. To prove the lemma, 
it suffices to show that the right-hand expression is equal to the left 
for $A = (\alpha'_1, \alpha'_2,\cdots, \alpha'_{\gamma(u)})$. Since 
$D^{\ast}_u$ has a path to any vertex $\alpha'_k \in A$, the edge-induced 
subgraph by $D^{\ast}_u$ is one connecting all vertices in $A \cup \{u, \alpha\}$.
That is, it is a feasible solution for instance 
$(G, E_A \cup \{\{u, \alpha\}\})$, and thus we have $|D^{\ast}_u| 
\geq c(D^{\mathit{Opt}}(G, E_A \cup \{\{u, \alpha\}\}))$. Combining
the optimality of $D^{\ast}(u_k, \alpha'_k)$ for any
$u_k \in \mathrm{Child}(u)$, we can conclude that the right-hand
is equal to the optimal cost $c(D^{\ast}(u, \alpha))$.
\qed
\end{proof}

\end{document}